\def\UseBibLatex{1}

\makeatletter
\def\input@path{{styles/}}
\makeatother

\documentclass[11pt]{article}

\providecommand{\BibLatexMode}[1]{}
\providecommand{\BibTexMode}[1]{}

\renewcommand{\BibLatexMode}[1]{#1}
\renewcommand{\BibTexMode}[1]{}

\ifx\UseBibLatex\undefined%
  \renewcommand{\BibLatexMode}[1]{}
  \renewcommand{\BibTexMode}[1]{#1}
\fi

\BibLatexMode{%
   \usepackage[bibencoding=utf8,style=alphabetic,backend=biber]{biblatex}%
   \usepackage{sariel_biblatex}%
}

\usepackage[cm]{fullpage}%
\usepackage{amsmath}%
\usepackage{amssymb}%
\usepackage[table]{xcolor}%

\usepackage[amsmath,thmmarks]{ntheorem}%

\usepackage{titlesec}%
\usepackage{xcolor}%
\usepackage{mleftright}%
\usepackage{xspace}%
\usepackage{graphicx}
\usepackage{hyperref}%
\usepackage[inline]{enumitem}
\usepackage{hyperref}%
\usepackage[ocgcolorlinks]{ocgx2}

\usepackage{salgorithm}

\hypersetup{%
      unicode,
      breaklinks,%
      colorlinks=true,%
      urlcolor=[rgb]{0.25,0.0,0.0},%
      linkcolor=[rgb]{0.5,0.0,0.0},%
      citecolor=[rgb]{0,0.2,0.445},%
      filecolor=[rgb]{0,0,0.4},
      anchorcolor=[rgb]={0.0,0.1,0.2}%
}

\titlelabel{\thetitle. }%

\theoremseparator{.}%

\theoremstyle{plain}%
\newtheorem{theorem}{Theorem}[section]

\newtheorem{lemma}[theorem]{Lemma}

\theoremstyle{plain}%
\theoremheaderfont{\sf} \theorembodyfont{\upshape}%
\newtheorem*{remark:unnumbered}[theorem]{Remark}%

\newtheorem{defn}[theorem]{Definition}

\theoremheaderfont{\em}%
\theorembodyfont{\upshape}%
\theoremstyle{nonumberplain}%
\theoremseparator{}%
\theoremsymbol{\myqedsymbol}%
\newtheorem{proof}{Proof:}%

\providecommand{\emphind}[1]{}%
\renewcommand{\emphind}[1]{\emph{#1}\index{#1}}

\definecolor{blue25emph}{rgb}{0, 0, 11}

\providecommand{\emphic}[2]{}
\renewcommand{\emphic}[2]{\textcolor{blue25emph}{%
      \textbf{\emph{#1}}}\index{#2}}

\providecommand{\emphi}[1]{}%
\renewcommand{\emphi}[1]{\emphic{#1}{#1}}

\definecolor{almostblack}{rgb}{0, 0, 0.3}

\providecommand{\emphw}[1]{}%
\renewcommand{\emphw}[1]{{\textcolor{almostblack}{\emph{#1}}}}%

\providecommand{\emphOnly}[1]{}%
\renewcommand{\emphOnly}[1]{\emph{\textcolor{blue25emph}{\textbf{#1}}}}

\newcommand{\myqedsymbol}{\rule{2mm}{2mm}}

\newcommand{\HLink}[2]{\hyperref[#2]{#1~\ref*{#2}}}
\newcommand{\HLinkSuffix}[3]{\hyperref[#2]{#1\ref*{#2}{#3}}}

\newcommand{\figlab}[1]{\label{fig:#1}}
\newcommand{\figref}[1]{\HLink{Figure}{fig:#1}}

\newcommand{\thmlab}[1]{{\label{theo:#1}}}
\newcommand{\thmref}[1]{\HLink{Theorem}{theo:#1}}

\newcommand{\lemlab}[1]{\label{lemma:#1}}
\newcommand{\lemref}[1]{\HLink{Lemma}{lemma:#1}}%

\providecommand{\eqlab}[1]{}%
\renewcommand{\eqlab}[1]{\label{equation:#1}}

\providecommand{\remove}[1]{}%
\newcommand{\Set}[2]{\left\{ #1 \;\middle\vert\; #2 \right\}}

\newcommand{\pth}[1]{\mleft(#1\mright)}%

\newcommand{\ceil}[1]{\mleft\lceil {#1} \mright\rceil}
\newcommand{\floor}[1]{\mleft\lfloor {#1} \mright\rfloor}

\renewcommand{\th}{th\xspace}

\newcommand{\T}{{\mathcal{T}}}

\newcommand{\G}{{\mathsf{G}}}
\newcommand{\ZZ}{\mathbb{Z}}%

\renewcommand{\Re}{\mathbb{R}}%

\newlist{compactenumA}{enumerate}{5}%
\setlist[compactenumA]{topsep=0pt,itemsep=-1ex,partopsep=1ex,parsep=1ex,%
   label=(\Alph*)}%

\newlist{compactenuma}{enumerate}{5}%
\setlist[compactenuma]{topsep=0pt,itemsep=-1ex,partopsep=1ex,parsep=1ex,%
   label=(\alph*)}%

\newlist{compactenumI}{enumerate}{5}%
\setlist[compactenumI]{topsep=0pt,itemsep=-1ex,partopsep=1ex,parsep=1ex,%
   label=(\Roman*)}%

\newlist{compactenumi}{enumerate}{5}%
\setlist[compactenumi]{topsep=0pt,itemsep=-1ex,partopsep=1ex,parsep=1ex,%
   label=(\roman*)}%

\newlist{compactitem}{itemize}{5}%
\setlist[compactitem]{topsep=0pt,itemsep=-1ex,partopsep=1ex,parsep=1ex,%
   label=\ensuremath{\bullet}}%

\numberwithin{figure}{section}%
\numberwithin{table}{section}%
\numberwithin{equation}{section}%

\newcommand{\Area}{{\mathrm{Area}}}
\newcommand{\Ropt}{R_{\mbox{\scriptsize \rm opt}}}
\newcommand{\Grid}{{\mathrm{Grid}}}
\newcommand{\Diam}{{\cal D}}
\def\SCH{{\cal S}}
\def\SCHX{{{\cal S}^{\oplus}}}
\newcommand{\Width}{\mathrm{Width}}
\def\qed{\hfill{\hfill\rule{2mm}{2mm}}}
\newcommand{\approxMVBB}{{\sc GridSearchMinVolBbx}}
\newcommand{\BG}{{\mathsf{BG}}}
\newcommand{\Od}{{\small\bf od}\ }
\def\M{{\mathsf{M}}}
\newcommand{\Bopt}{B_{\mbox{\scriptsize \rm opt}}}
\newcommand{\eps}{\varepsilon}
\newcommand{\C}{\mathcal{C}}
\newcommand{\B}{\mathcal{B}}
\newcommand{\A}{\mathcal{A}}
\newcommand{\V}{\mathcal{V}}
\newcommand{\CH}{{\mathcal{CH}}}

\newcommand{\Vol}{\mathrm{Vol}}

\newcommand{\seclab}[1]{\label{sec:#1}}
\newcommand{\secref}[1]{\HLink{Section}{sec:#1}}
\newcommand{\tbllab}[1]{\label{table:#1}}
\newcommand{\tblref}[1]{\HLink{Table}{table:#1}}

\BibLatexMode{%
   \bibliography{mvbb}
}

\begin{document}
\title{Efficiently Approximating the Minimum-Volume Bounding Box of a Point Set in Three Dimensions%
   \thanks{%
      Work on this paper by the first author was done while he was affiliated with the Center for Geometric Computing, Dept.\ of Computer Science, Johns Hopkins University, Baltimore, MD 21218.  A preliminary version of this paper appeared in \cite{bh-eamvb-99-conf}, and the journal version appeared in \cite{bh-eamvb-01} Work of the first author has been supported by the U.S.  Army Research Office under Grant DAAH04-96-1-0013.  Work of the second author has been supported by a grant from the U.S.-Israeli Binational Science Foundation.  }  }

\author{%
   Gill Barequet\thanks{Faculty of Computer Science, The Technion---IIT, Haifa 32000, Israel}%
   \and%
   Sariel Har-Peled\thanks{Department of Computer Science, DCL 2111; University of Illinois; 1304 West Springfield Ave.; Urbana, IL 61801; USA} }

\date{June 30, 2001}
\maketitle

\begin{abstract}
    We present an efficient $O (n + 1/\eps^{4.5})$-time
    algorithm for computing a $(1+\eps$)-approximation of
    the minimum-volume bounding box of $n$ points in
    $\Re^3$.  We also present a simpler algorithm (for the
    same purpose) whose running time is $O (n \log{n} + n /
    \eps^3)$.  We give some experimental results with
    implementations of various variants of the second
    algorithm. The implementation of the algorithm described
    in this paper is available online \cite{h-scpca-00}.
\end{abstract}

\section{Introduction}

In this paper we give efficient algorithms for solving the following
problem:
\begin{quote}
   Given a set $S$ of $n$ points in $\Re^3$ and a parameter
   $0 < \eps \leq 1$, find a box that encloses $S$ and approximates
   the minimum-volume bounding box of $S$ by a factor $(1 + \eps)$.
\end{quote}
We are not aware to any previously-published algorithm that solves
this problem.

Three-dimensional boxes which enclose sets of points are used for
maintaining hierarchical partitioning of sets of points. These data
structures have important applications in computer graphics (e.g., for
fast rendering of a scene or for collision detection), statistics (for
storing and performing range-search queries on a large database of
samples), etc. From a top-down viewpoint, the problem in such
applications is divided into two (admittedly, related) problems of
splitting a given set of points into two (or more) subsets, and of
computing a nearly-optimal box (or another generic shape) that
encloses each subset. In this paper we concentrate on the second
problem.

Numerous heuristics have been proposed for computing a box which encloses a given set of points.  The simplest heuristic was naturally to compute the axis-aligned bounding box (AABB) of the point set.  Two-dimensional variants of this heuristic include the well-known R-tree, the {\em packed\/} R-tree~\cite{rl-dsspd-85}, the R$^+$-tree~\cite{srf-rtdim-87}, the R$^*$-tree~\cite{bkss-trter-90}, etc.  \cite{hkm-ecdmv-95} use a minimum-volume AABB trimmed in a fixed number of directions for speeding up collision detection.  \cite{glm-othsr-96} implement in their RAPID system the OBB-tree (a tree of arbitrarily-oriented bounding boxes), where each box which encloses a set of polygons is aligned with the {\em principal components\/} of the distribution of polygon vertices.  A similar idea is used by \cite{bcgmt-bthrs-96} for the BOXTREE.  The latter work suggests a few variants, in which the computed box is aligned with only one principal component of the point distribution (e.g., the one that corresponds to the smallest or largest eigenvalue of the covariance matrix of the point coordinates), and the other two directions are determined by another method (e.g., by computing the exact minimum-area bounding rectangle of the projection of the points into a plane orthogonal to the first chosen direction).  Other generic shapes, such as a sphere~\cite{h-cdiga-95}, a cone~\cite{s-dasds-90}, or a prism~\cite{fp-ochr3-87,bcgmt-bthrs-96} were also used for maintaining a hierarchical data structure of point containers.  Most of these heuristics require $O (n)$ time and space for computing the bounding box (or another shape) but do not provide a guaranteed value (approximation factor of the optimum) of the output.

An algorithm of \cite{bs-cspta-97} solves a similar problem,
in which the $n$ points are to be contained in two axis-aligned boxes,
and the goal is to minimize the volume (or any other monotone measure)
of the larger box. Their algorithm requires $O (n^2)$ time.

O'Rourke presented the only algorithm (to the best of our knowledge)
for computing the exact arbitrarily-oriented minimum-volume bounding
box of a set of $n$ points in $\Re^3$. His algorithm requires
$O (n^3)$ time and $O (n)$ space.
In this paper we present the first two $(1 + \eps)$-approximation
algorithms for the minimum-volume bounding-box problem. Both algorithms
are nearly linear in $n$.  The running times of these algorithms are
$O (n + 1/\eps^{4.5})$ and $O (n \log n + n / \eps^3)$.

The paper is organized as follows.  In \secref{notations} we give the notations and definitions used throughout the paper.  In \secref{first-alg} and~\secref{second-alg} we present the two $(1 + \eps)$-approximation bounding-box approximations algorithms.  In \secref{experiments} we presents experimental results.  We end in \secref{conclusion} with some concluding remarks.

\section{Notations and Definitions}
\seclab{notations}

We first present the notation used in this paper.

We denote the origin of coordinates by $o$ and the unit cube by $\C$.
Throughout the paper we denote by $Q$ a point set in two dimensions,
and use $S$ to denote a similar set in three dimensions.
Unless specified otherwise, the set is assumed to contain $n$ (or a
comparable number of) points.
We denote by $\CH (Q)$ (resp., $\CH (S)$) the 2-dimensional (resp.,
3-dimensional) convex hull of $Q$ (resp., $S$).
The symbol $P$ is used for denoting a convex polyhedron in three-space.

The symbols $R$ and $B$ are used for rectangles and boxes, respectively.
The notation $B (S)$ is used for any bounding box of a point set $S$.
We also use the notation $B = (b_1, b_2, b_3)$, where $b_1, b_2, b_3$
are three orthogonal vectors in $\Re^3$ having the directions and
sizes of the edges of $B$.
The operators $\Area (R)$ and $\Vol (B)$ denote the area of $R$ and the
volume of $B$, respectively.

We denote by $\Ropt (Q)$ the {\em minimum-area bounding rectangle\/} of $Q$ and by $\Bopt (S)$ the {\em minimum-volume bounding box\/} of $S$.\footnote{Since we are interested in the {\em area\/} (or {\em volume}) of the object, we may pick an arbitrary minimum rectangle (or box) if several exist.  } Let $V$ be a set of orthogonal vectors in $\Re^3$.  We denote by $\Bopt (S,V)$ for the minimum-volume bounding box of $S$ whose set of directions contains $V$.  We denote some constant approximation of the minimum-volume bounding box of $S$ (for some predefined positive constant) by $B^* (S)$.  (Such a box is computed in \lemref{bounding:fifteen} and \lemref{bound:fifteen:ext}---see \secref{first-alg}.)

Finally, given a box $B = (b_1, b_2, b_3)$ in $\Re^3$,
the grid of points spanned by $B$ is
\begin{equation*}
    \Grid (B) = \Set{i_1 b_1 + i_2 b_2 + i_3 b_3}{i_1, i_2, i_3 \in \ZZ}.
\end{equation*}
Let $\G = \Grid (B)$. Denote by
\begin{equation*}
    B^\G_{(i,j,k)}
    =
    \Set{x_1 b_1 + x_2 b_2 + x_3 b_3}
    {%
       \begin{array}{c}
         i \leq x_1 \leq i + 1, \\
         j \leq x_2 \leq j + 1, \\
         k \leq x_3 \leq k + 1, \\
       \end{array}
       i,j,k \in \ZZ
    }
\end{equation*}
the $(i,j,k)$\th cell of $\G$.  For a prescribed constant $d >0$, let
\begin{equation*}
    G (B, d)
    =
    \Set{i_1 b_1 + i_2 b_2 + i_3 b_3}%
    {i_1, i_2, i_3 \in \ZZ, |i_1|, |i_2|, |i_3| \leq d}
\end{equation*}
be the set of points of $\G$ whose $L_\infty$-distance (along the
directions $b_1, b_2, b_3$) from $o$ is at most $d$.

\section{An Efficient Approximation Algorithm}

\seclab{first-alg}

In this section we present our main approximation algorithm to
the minimum-volume bounding box of a set of points in three
dimensions.

\subsection{Approximating the Diameter}

First we need a good-enough approximation of the diameter of the point set.

\begin{defn}
   The {\em diameter\/} of a point set $S \in \Re^3$ (denoted
   by $\Diam(S)$) is the distance between the two furthest points of
   $S$. That is, $\Diam(S) = max_{s,t \in S} |st|$.
\end{defn}

We can easily find a pair of points in $S$ whose mutual distance is a
$1 / \sqrt{d}$-approximation of the diameter of $S$:\footnote{
   Note that we approximate the diameter from {\em below}, while we
   approximate the minimum-volume bounding box from {\em above}.
}

\begin{lemma}
    \lemlab{constant:diam}

   Given a point set $S$ in $\Re^d$ (for a fixed $d$), one can compute in
   $O (n)$ time a pair of points $s,t \in S$,
   such that $|st| \leq \Diam (S) \leq \sqrt{d} |st|$.
\end{lemma}

\begin{proof}
   Let $B$ be the minimum axis-parallel box containing $S$, and let
   $s$ and $t$ be the points in $S$ that define the longest edge of $B$,
   whose length is denoted by $l$.
   By the diameter defn, $|st| \leq \Diam (S)$,
   and clearly, $\Diam (S) \leq \sqrt{d} \, l \leq \sqrt{d} |st|$.
   The points $s$ and $t$ are easily found in $O (n d)$ time.
   Since $d$ is fixed, this is actually $O (n)$ time.
\end{proof}

In particular, we can approximate in linear time the diameter of a
point set in three dimensions by a factor of $1 / \sqrt{3}$.
(In fact, one can find in linear time a $(1 / \sqrt{3})$-approximation
of the diameter in {\em any\/} dimension; see~\cite{ek-adspe-89}.)

Actually, we can find any arbitrarily-good approximation of the diameter.
Since we were not able to find any reference with a proof of the following
folklore lemma (see~\cite{h-aspgd-99} for a similar result),
we include an easy proof of it here.

\begin{lemma}
   \lemlab{diam:eps}

   Given a set $S$ of $n$ points in $\Re^d$ (for a fixed $d$) and
   $\eps > 0$, one
   can compute in $O(n + 1/\eps^{2(d-1)})$ time a pair of points
   $s,t \in S$ such that $|st| \geq (1 - \eps) \Diam (S)$.
\end{lemma}

\begin{proof}
    Let $B$ be the minimum axis-parallel box containing $S$, and let $\G = \Grid ((\eps / (2 \sqrt{d})) B)$.  For a point $x \in \Re^d$, denote by $x_\G$ the closest point of $\G$ to $x$. Define $S_\G = \Set{x_\G }{x \in S}$.  Finally, let $l$ be the length of the longest diagonal of a cell of $\G$. Clearly, $l \leq (\eps / 2) \Diam (S)$.  For every pair of points $x,y \in S$ we have
   $$
      |xy| - l \leq |x_\G y_\G| \leq |xy| + l.
   $$
   Thus, $\Diam (S_\G) \geq \Diam (S) - l$.
   Let $s,t$ be the two points in $S$ for which $s_\G, t_\G$ realize the
   diameter of $S_\G$. Now,
   $$
      |st| \geq |s_\G t_\G| - l = \Diam (S_\G) - l \geq \Diam (S) - 2 l
         \geq (1 - \eps) \Diam (S).
   $$

   The set $S_\G$ can be computed in $O (n)$ time (where the hidden
   constant of proportionality contains a factor $d$).  The cardinality
   of $S_\G$ is $O (1 / \eps^d)$. Thus, we can compute the diameter of
   $S_\G$ (in a brute force manner) in $O (1 / \eps^{2d})$ time.
   The two points of $S$ that correspond to the diameter of $S_\G$
   are, by the above analysis, the sought $(1 - \eps)$-approximation
   of the diameter of $S$.

   Any point of $S_\G$ that lies between two other points of $S_\G$
   (along one of the axes) can not correspond to a point of $S$ that
   defines $\Diam (S)$.
   By removing all such points in $O (n)$ time we can consider only
   $O (1 / \eps^{d-1})$ points.
   Hence the running time of the algorithm is improved to
   $O (n + 1 / \eps^{2 (d-1)})$.
\end{proof}

The two points of $S_\G$ that realize its diameter must be vertices of the convex-hull of $S_\G$.  Thus, the running time of the algorithm of \lemref{diam:eps} can be further improved by first computing the set of vertices of $\CH (S_\G)$, denoted by $S_{\CH(S_\G)}$.  Set $h=|S_{\CH(S_\G)}|$.  \cite{a-lbvsc-63} showed that $h = O (1 / \eps^{(d-1)d/(d+1)})$.  For $d \leq 3$ we compute $S_{\CH(S_\G)}$ by computing the entire convex-hull of $S_\G$ in $O (|S_\G| \log |S_\G|)$ time.  For higher dimensions we use an output-sensitive algorithm~\cite{c-osrch-96}.  Let $m$ denote the number of points in $S_\G$.  Clearly, $m = O (1 / \eps^{(d-1)})$.  The time required for computing $S_G$ is
$$
   O \pth{ m \log^{d+2} h + (mh)^{1-\frac{1}{\floor{d/2}+1} } \log^{O(1)} m }
      = O \pth{ m^{\frac{2d}{d+1}}}
      = O \pth{\pth{\frac{1}{\eps}}^{\frac{2d(d-1)}{d+1}}}.
$$
Computing the diameter of $S_{\CH(S_\G)}$ (in a brute-force manner)
requires $O(h^2) = O\pth{\pth{{1}/{\eps}}^{\frac{2d(d-1)}{d+1}}}$ time.
Overall, the running time of the algorithm of
\lemref{diam:eps} can be improved to
$O\pth{\pth{{1}/{\eps}}^{\frac{2d(d-1)}{d+1}}}$.

We can do even better in three dimensions if we are willing to sacrifice
simplicity. In this case we compute the exact diameter of $\CH (S_\G)$
in $O((1 / \eps^{3/2}) \log{(1/\eps)})$ time~\cite{cs-arscg-89}.
Overall, we compute a $(1 - \eps)$-approximation of the diameter of
$S \in \Re^3$ in $O (n + (1 / \eps^{3/2}) \log{(1/\eps)})$ time.

\subsection{Computing an Approximating Box}

Let $Q$ be a set of $n$ points in $\Re^2$.  Computing $\Ropt (Q)$ can
be done in $O (n \log{n})$ time~\cite{t-sgprc-83}.  (Hence, given a set
$S$ of $n$ points and a direction $v$ in $\Re^3$, one can compute
$\Bopt(S,\{v\})$ in $O (n \log{n})$ time.)  The bottleneck of the cited
algorithm is the computation of $\CH (Q)$; when the latter is given in
advance, $\Ropt (Q)$ can be computed in optimal $\Theta (n)$ time.

\begin{lemma}
   \lemlab{bounding:fifteen}

   Given a set $S$ of $n$ points in $\Re^3$, one can compute in
   $O (n)$ time a bounding box $B (S)$ with
   $\Vol (\Bopt (S)) \leq \Vol (B (S)) \leq
      6 \sqrt{6} \, \Vol (\Bopt (S))$.
\end{lemma}

\begin{proof}
   By using the algorithm of \lemref{constant:diam} we
   compute in $O (n)$ time two points $s,t \in S$ which form a
   $(1 / \sqrt{3})$-approximation of the diameter of $S$.
   Let $H$ be a plane perpendicular to $st$, and let $Q$ be the
   orthogonal projection of $S$ into $H$.

   Now, by using the algorithm of \lemref{constant:diam} again, we compute in $O (n)$ time two points $s',t' \in Q$ (see \figref{2d-br}) for which $\sqrt{2} |s't'| \geq \Diam (Q)$.  Let $\mu$ be a direction perpendicular to $st$ and $s't'$.  We claim that the box $B^* = \Bopt (S, \{st, s't', \mu\})$ is a $(6 \sqrt{6})$-approximation of $\Bopt (S)$.

   Indeed, let $R$ be the minimum-area bounding rectangle of $Q$ in the
   directions $s't'$ and $\mu$, let $\omega$ be the length of the edge of
   $R$ in the direction $\mu$, and let $u,v$ be the two points of $Q$
   lying on the two edges of $R$ parallel to $s't'$.  Clearly,
   $$
      \Area (R) \leq \omega \Diam (Q) \leq \sqrt{2} \, \omega |s't'|.
   $$
   On the other hand, the quadrilateral $F = s'ut'v$ is contained in
   $\CH (Q)$, and its area is $|s't'| \omega / 2$.

   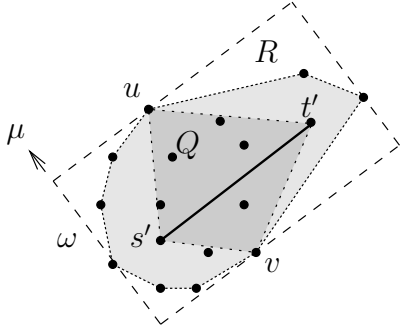
\begin{figure}
      \centering
      \begin{tabular}{c}
\begin{picture}(0,0)%
\includegraphics{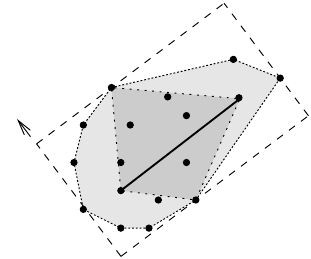}
\end{picture}%
\setlength{\unitlength}{3947sp}%
\begingroup\makeatletter\ifx\SetFigFont\undefined
\def\x#1#2#3#4#5#6#7\relax{\def\x{#1#2#3#4#5#6}}%
\expandafter\x\fmtname xxxxxx\relax \def\y{splain}%
\ifx\x\y   %
\gdef\SetFigFont#1#2#3{%
  \ifnum #1<17\tiny\else \ifnum #1<20\small\else
  \ifnum #1<24\normalsize\else \ifnum #1<29\large\else
  \ifnum #1<34\Large\else \ifnum #1<41\LARGE\else
     \huge\fi\fi\fi\fi\fi\fi
  \csname #3\endcsname}%
\else
\gdef\SetFigFont#1#2#3{\begingroup
  \count@#1\relax \ifnum 25<\count@\count@25\fi
  \def\x{\endgroup\@setsize\SetFigFont{#2pt}}%
  \expandafter\x
    \csname \romannumeral\the\count@ pt\expandafter\endcsname
    \csname @\romannumeral\the\count@ pt\endcsname
  \csname #3\endcsname}%
\fi
\fi\endgroup
\begin{picture}(2467,2049)(547,-1724)
\put(2167,-1378){\makebox(0,0)[lb]{\smash{\SetFigFont{12}{14.4}{rm}$v$}}}
\put(547,-538){\makebox(0,0)[lb]{\smash{\SetFigFont{12}{14.4}{rm}$\mu$}}}
\put(1280,-283){\makebox(0,0)[lb]{\smash{\SetFigFont{12}{14.4}{rm}$u$}}}
\put(2394,-406){\makebox(0,0)[lb]{\smash{\SetFigFont{12}{14.4}{rm}$t'$}}}
\put(1321,-1223){\makebox(0,0)[lb]{\smash{\SetFigFont{12}{14.4}{rm}$s'$}}}
\put(1607,-630){\makebox(0,0)[lb]{\smash{\SetFigFont{12}{14.4}{rm}$Q$}}}
\put(864,-1224){\makebox(0,0)[lb]{\smash{\SetFigFont{12}{14.4}{rm}$\omega$}}}
\put(2101,-61){\makebox(0,0)[lb]{\smash{\SetFigFont{12}{14.4}{rm}$R$}}}
\end{picture}

      \end{tabular}
      \caption{The convex hull of the  projection of $S$ contains a large
               quadrangle}
      \figlab{2d-br}
   \end{figure}

   Set $\SCH = \CH (S)$. Let $s_H$ be the orthogonal projection of $s$ into
   $H$, and let $T_1, \ldots, T_m$ be a triangulation of $\CH (Q)$ (within $H$)
   in which all the triangles share the vertex $s_H$.

   We next show that $\Vol (\SCH) \geq \Area(\CH (Q)) |st| / 3$.  The outer
   edge $p_i' q_i'$ of $T_i$ is a projection of a pair of points
   $p_i, q_i \in S$ (for $i = 1, \ldots, m$) into $H$.
   Thus, each triangle $T_i$ corresponds to a tetrahedron
   $\T_i = \CH \pth{\{ p_i, q_i, s, t\}}$ (for $i = 1, \ldots, m$).
   These tetrahedra are pairwise disjoint in their interiors, and
   $$
   \Vol (\T_i) = \Vol (\CH (\{p_i, q_i ,s ,t\})) = \frac{\Area (T_i) |st|}{3},
   $$
   where we use the fact that the volume of a tetrahedron does not
   change when we translate one of its vertices $u$ in a direction
   parallel to an edge (of the tetrahedron) that does not coincide
   with $u$ (or, in more generality, parallel to the opposite face
   of $u$).  We thus have
   \begin{eqnarray*}
       \lefteqn{ \hspace{-1cm}
       \frac{|s't'| |st| \omega}{6}
          =    \frac{\Area(F) |st|}{3}
          \leq \frac{\Area (\CH (Q)) |st|}{3} } \\
       & &
          =    \frac{(\sum_{i=1}^{m} \Area (T_i)) |st|}{3}
          =    \sum_{i=1}^{m} \Vol (\T_i)
          \leq \Vol (\SCH).
   \end{eqnarray*}

   On the other hand,
   $$
      \Vol (\SCH) \leq \Vol (B^*) \leq \Area (R) \Diam(S) \leq
         \sqrt{2} \omega |s't'| \sqrt{3} |st|.
   $$
   Since $S \subseteq B^*$, we have
   $\Vol (\SCH) \leq \Vol(B^*) \leq 6 \sqrt{6} \, \Vol(\SCH)$,
   as required.
\end{proof}

Note that the approximation factor for the diameter obtained in
\lemref{bounding:fifteen} can be improved by using better
approximation schemes, e.g., the algorithm of \lemref{diam:eps}.
In any case a better constant approximation factor will be reflected by a
higher constant hidden in the big-Oh notation of the running time.

The algorithm of \lemref{bounding:fifteen} can be extended to a $d$-dimensional space by choosing the direction of the exact diameter $v$ of the point set $S$ as one direction of the bounding box, projecting the points to a hyperplane $H \in \Re^{d-1}$ perpendicular to $v$, solving recursively the $(d-1)$-dimensional problem in $H$, and outputing the Cartesian product of the $(d-1)$-dimensional solution and $v$.  The volume of the computed box is at most $d!$ times the volume of the optimal (minimum-volume) bounding box of $S$.\footnote{ Here we can trade time for approximation quality.  By investing less time we can compute a $(1 / \sqrt{2})$- (resp., $(1 / \sqrt{3})$-) approximation of the diameter in two (resp., $k$, for $3 \leq k \leq d$) dimensions, and obtain a $(\sqrt{2} \cdot 3^{(d-2)/2} d!)$-approximation of $\Bopt (S)$.  By investing even less time we can compute a $(1 / \sqrt{k})$-approximation of the diameter in $k$ dimensions (for $2 \leq k \leq d$), and obtain a $(d!)^{3/2}$-approximation of $\Bopt (S)$.  } Moreover, the side lengths of the bounding box are found in decreasing order. \cite{fl-fmfai-95} use a similar method for visualizing a set of points in a high-dimensional space. They compute in a similar manner the first two (or three) directions, project the points into the subspace spanned by these directions, and display the projected points in this subspace.

We next show that a large-enough convex polyhedron contained in the
unit cube $\C$ must contain a large-enough axis-parallel cube.

\begin{lemma}
   \lemlab{obvious}

   Let $P$ be a convex polyhedron with the properties that
   $P \subseteq \C$ and $\Vol (P) \geq 1/15$.
   Then there exists a translation $v \in \Re^3$ for which
   $\frac{1}{107} \C + v \subseteq P$.
\end{lemma}

\begin{proof}
    It is easy to verify that the area of the intersection of
    $\C$ with any plane is at most $3 \pi / 4$.
    (Indeed, the intersection area is maximized when the plane passes
    through the center of $\C$, and such a cross-section is contained
    in a disk of radius $\sqrt{3} / 2$.) Hence,
    $$
       \Width (P) \geq
          \Vol (P) / \pth{ \frac{3 \pi}{4} } \geq
          \frac{4}{45 \pi},
    $$
    where $\Width (P)$ is the minimum distance between two
    parallel planes supporting $P$.
    Let $r (P)$ be the radius of the largest ball $K$ inscribed in $P$.
    It is known~\cite{gk-iojrc-92} that
    $r (P) \geq \Width (P) / (2 \sqrt{3})$.
    This implies $r (P) \geq 2 / (45 \sqrt{3} \pi)$.
    Finally, $K$ inscribes an axis-parallel cube $\C'$ whose side is of
    length $(2 / \sqrt{3}) r (P) \geq 1/107$,
    and $\C' \subseteq K \subseteq P$, as asserted.
\end{proof}

Let $A$ and $B$ be two sets in $\Re^3$. The {\em Minkowski sum\/} of $A$ and $B$ is the set $A \oplus B = \Set{a + b }{a \in A, b \in B}$.
Given a constant $c > 0$ and a box $B \in \Re^3$, it is easily verified that
$B \oplus (c \cdot B)$ is also a box whose volume is $(1 + c)^3 \Vol (B)$.

The idea of our main algorithm is approximating $\CH (S)$
by a low-complexity convex polyhedron $P \supseteq \CH (S)$,
followed by computing (exactly) $\Bopt (P)$.
The polyhedron $P$ is chosen such that $\Bopt (P)$ is an
$(1 + \eps)$-approximation of $\Bopt (S)$.

First we need a combined version of \lemref{bounding:fifteen} and~\lemref{obvious}:

\begin{lemma}
   \lemlab{bound:fifteen:ext}

   Let $S$ be a set of $n$ points in $\Re^3$. One can compute in
   $O (n)$ time a bounding box $B^* (S)$ such that
   $\Vol (\Bopt (S)) \leq \Vol (B^* (S)) \leq 15 \, \Vol (\Bopt (S))$
   and there exists a translation $v \in \Re^3$ for which
   $\frac{1}{107} B^* (S) + v \subseteq \CH (S)$.
\end{lemma}

\begin{proof}
   Let $B = B^* (S)$ be the bounding box of $S$ computed by
   \lemref{bounding:fifteen}. Let $T$ and $t$ be a nonsingular
   linear transformation and a translation, respectively, such that
   $T (B) + t = \C$.
   Careful observation of the construction in
   \lemref{bounding:fifteen} shows that
   $\Vol (T (\CH (S))) \geq \frac{1}{15}$.
   Hence by \lemref{obvious} there exists a translation $v'$
   such that $\frac{1}{107} (T(B) + t) + v' \subseteq T (\CH (S)) + t$,
   where we use the fact that linear transformations
   preserve volume order.  Therefore
   $\frac{1}{107} B + T^{-1} (v' - \frac{106}{107} t) \subseteq \CH(S)$,
   and the claim follows.
\end{proof}

Note that this lemma is true even if $S$ is a planar set, in which
case the minimum-area bounding rectangle provided by
\lemref{bounding:fifteen} degenerates to a segment, and the
volume of the computed bounding box is 0. This guarantees that the
approximation algorithm (described below) produces a degenerate box
(of volume 0) for a degenerate (planar) input point set.

We are now ready to present the approximation algorithm for $\Bopt (S)$.

Let $B = B^* (S)$ be the bounding box of $S$ computed by \lemref{bound:fifteen:ext}, and let $B_\eps$ be a translated copy of $\frac{\eps}{428} B$ centered at $o$.  In addition, define $\SCHX = \CH (S) \oplus B_{\eps}$ and $\G = \Grid (\frac{1}{2} B_\eps)$.  We approximate $S$ on $\G$. For each point $p \in S$ let $\G (p)$ be the set of eight vertices of the cell of $\G$ that contains $p$, and let $S_\G = \cup_{p \in S} \G(p)$. Define $P = \CH (S_\G)$.  Clearly, $\CH (S) \subseteq P \subseteq \SCHX$. Moreover, one can compute $P$ in $O (n + (1/\eps^2) \log{(1/\eps)})$ time.  On the other hand, $P \subseteq B \oplus B_\eps$.  The latter term is a box which contains at most $k = 860 / \eps + 4$ grid points along each of the directions set by $B$, so $k$ is also an upper bound for the number of grid points contained by $P$ in each direction.  \cite{a-lbvsc-63} showed that the complexity of $P$ is $O (k^{3/2}) = O (1 / \eps^{3/2})$.  We exploit this result in the analysis of the running time of the algorithm.  Finally, we compute $\Bopt (P)$ exactly.

It remains to show that $\Bopt (P)$ is a $(1 + \eps)$-approximation of
$\Bopt (S)$.
Let $\Bopt^{\eps}$ be a translation of $\frac{\eps}{4} \Bopt (S)$
that contains $B_\eps$. (The existence of $\Bopt^{\eps}$ is guaranteed by
\lemref{bound:fifteen:ext}.) Thus,
$P \subseteq \CH (S) \oplus B_\eps \subseteq \CH (S) \oplus \Bopt^\eps
    \subseteq \Bopt (S) \oplus \Bopt^\eps$. Since
$\Bopt (S) \oplus \Bopt^\eps$ is a box, it is a bounding box of $P$
and therefore also of $\CH (S)$. Its volume is
$$
   \Vol (\Bopt (S) \oplus \Bopt^\eps) =
      \pth{1 + \frac{\eps}{4}}^3 \Vol(\Bopt (S))
   < (1 + \eps) \, \Vol (\Bopt (S)),
$$
as desired. (The last inequality is the only place where we use the
assumption $\eps \leq 1$.)

To recap, the algorithm consists of the four following steps:
\begin{enumerate}
   \item Compute the box $B^* (S)$ (see \lemref{bound:fifteen:ext})
         in $O (n)$ time.

   \item Compute the point set $S_\G$ in $O (n)$ time.

   \item Compute $P = \CH (S_\G)$ in $O (n + (1/\eps^2) \log{(1/\eps)})$
         time.
         This is done by computing the convex hull of all the extreme points
         of $S_\G$ along vertical lines of $\G$. We have $O (1 / \eps^2)$
         such points, thus computing their convex hull takes
         $O ((1 / \eps^2) \log (1 / \eps))$ time.

         \item Compute $\Bopt (P)$ by the algorithm of \cite{o-fmeb-85}.  The complexity of $P$ is $O (1/\eps^{3/2})$~\cite{a-lbvsc-63}, so this step requires $O ((1 / \eps^{3/2})^3) = O (1 / \eps^{4.5})$ time.
\end{enumerate}

Thus we obtain our main result:

\begin{theorem}
    \thmlab{aprox}

   Let $S$ be a set of $n$ points in $\Re^3$, and let $0 < \eps \leq 1$ be a
   parameter. One can compute in $O (n  + 1 / \eps^{4.5})$ time a bounding
   box $B (S)$ with $\Vol (B (S)) \leq (1 + \eps) \, \Vol (\Bopt (S))$.
   \qed
\end{theorem}

Note that the box $B (S)$ computed by the above algorithm is most likely
not minimal along its directions. The minimum bounding box of $S$ homothet
of $B (S)$ can be computed in additional $O (n)$ time.

\section{An Alternative Practical Algorithm}

\seclab{second-alg}

Unfortunately, the algorithm described in the previous section is
perhaps too difficult to implement.  In this section we suggest an
asymptotically less efficient, but easier to implement, approximation
algorithm for the minimum-volume bounding box problem.

In the algorithm described above we chose an approximation of the diameter of the set $S$ as a ``favorable'' direction $v$ and computed $\Bopt (S, \{v\})$, which served for the definition of the grid $\G$.  Then we expanded $\CH (S)$ (which was not computed explicitly) into a low-complexity grid polyhedron $P$, and computed $\Bopt (P)$ exactly.  We will next show that some grid point $v^* \in \G$ is itself favorable in the sense that $\Bopt (S, \{o v^*\})$ is a $(1 + \eps)$-approximation of $\Bopt (S)$.  Furthermore, the point $v^*$ is close enough to the origin of $\G$ so that we can perform an exhaustive search for this point.  For this purpose we will compute $\CH (S)$ explicitly and output $\Bopt (S, \{o v^*\})$ as the sought approximation.

\begin{figure*}
   \begin{center}
      \small
      \begin{program}
          \> {\large{\sc{Algorithm}}} \
          \Proc{\approxMVBB\ ($S, \eps$)} \\
          \> \> {\tt Input:} \> \> \> A set $S$ of $n$ points in $\Re^3$, and
          a parameter $0 < \eps \leq 1$. \\
          \> \> {\tt Output:} \> \> \>
          A $(1 + \eps)$-approximation of $\Bopt (S)$. \\
          \> \Procbegin \\
          \> \> Compute $\CH (S)$; \\
          \> \> Compute $B^* (S)$; \> \> \> \> \> \> \verb+/*+ The box generated by
          \lemref{bound:fifteen:ext} \verb+*/+ \\
          \> \> Compute $\BG = G (B^* (S), c / \eps)$; \verb+/*+ Refer to the text for the value of $c$
          \verb+*/+ \\
          \> \> Set $\verb+min_vol+ := \infty$ and $v^* := \verb+undefined+$; \\
          \> \> \For $v \in \BG$ \Do \\
          \> \> \> Compute $B = \Bopt (S, \{v\})$; \\
          \> \> \> \If $\verb+min_vol+ > \Vol (B)$ \Then \Do \\
          \> \> \> \> Set $\verb+min_vol+ := \Vol (B)$ and
          $v^* := v$; \\
          \> \> \> \Od \\
          \> \> \Endfor \\
          \> \> Return $\Bopt (S, \{v^*\})$; \\
          \> \Endproc{\approxMVBB} \\
      \end{program}
   \end{center}
   \vspace{-0.8cm}
   \caption{An exhaustive grid-based search algorithm for approximating
            $\Bopt (S)$}
   \figlab{aprox:mvbb}
\end{figure*}

The approximation algorithm is depicted in \figref{aprox:mvbb}.  In a nutshell, the algorithm computes $B^* (S)$, builds the corresponding grid $\G$, and computes the box $\Bopt (S, \{v\})$ for all the grid points $v \in \G$ close enough to $o$.  The overall running time of the algorithm is $O (n \log n + n / \eps^3)$.  We next prove the correctness of the algorithm and (implicitly) compute the constant $c$ which it uses.

Let $B^*(S)$ be, as before, the bounding box of $S$ computed by \lemref{bound:fifteen:ext}.  Define the grid $\G = \Grid (c \eps B^* (S))$.  The set of directions induced by $\BG$ (points computed by the algorithm) is a finite subset of $\G$.  Let $m = \max \pth{\ceil{4 / \eps}, 6}$.  Also, let $B^{1/m} (S)$ be a translation of $\frac{1}{m} \Bopt (S)$ centered at $o$, and let $\B = \Bopt (S) \oplus \Bopt^{1/m} (S)$.  We assume without loss of generality that $\B$ is axis parallel, its minimum in all axes is $o$, and that the longest edge of $\B$ is parallel to the $z$ axis.  Set $\delta = 1 / (10 m + 20)$, and define a second grid $\M = \Grid (\delta \B)$. (The latter grid plays a role only in proving the correctness of the algorithm but not in the algorithm itself.)  \lemref{bound:fifteen:ext} tells us that the constant $c$ may be chosen small enough so as to make the interior of every cell of $\M$ contain at least one grid point of $\G$.

Is it fairly easy to prove that:

\begin{lemma}
   \lemlab{project:rects}
   Let $R$ and $H$ be a rectangle and a plane, respectively, in $\Re^3$.  There exists a rectangle $R' \subset H$ that contains the orthogonal projection of $R$ into $H$, such that $\Area (R') \leq \Area (R)$.  \qed
\end{lemma}

Our next goal is to show that there exist two $\G$-grid points, where the
minimum-volume bounding box of $S$ perpendicular to the direction defined
by these points is a good approximation of $\Bopt (S)$. (Note that in the
grid-search algorithm the first point is actually $o$.)

\begin{lemma}
    Let $p$ and $q$ be points of $\G \cap B^\M_{(0,0,5)}$ and $\G \cap B^\M_{(0,0,20m+15)}$, respectively.  Then $\Bopt (S, \{pq\})$ is a $(1 + \eps)$-approximation of $\Bopt (S)$.
\end{lemma}

\begin{proof}
    Refer to \figref{proof}.
   \begin{figure*}
      \centerline{\includegraphics{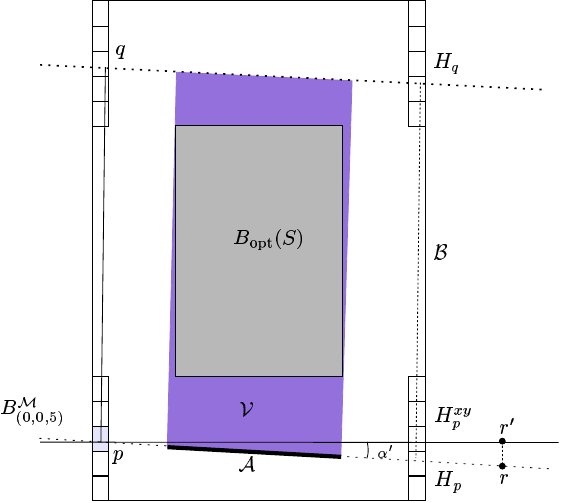}}
      \caption{The construction used in
         \thmref{aprox:ext}}
      \figlab{proof}
   \end{figure*}
   Denote by $H_p$ and $H_q$ the two planes that pass through $p$ and $q$,
   respectively, and are perpendicular to $pq$.
   Also denote the angle between $pq$ and the $z^+$ axis by $\alpha$.
   By construction, the direction $pq$ is almost vertical:
   A simple calculation shows that $\tan (\alpha) \leq \sqrt{2} / (10 m)$.
   Let $H_p^{xy}$ be the $xy$-parallel plane that passes through $p$,
   let $r \in H_p$ be an arbitrary point for which $|pr| \leq 2 |pq|$,
   and let $r'$ be the orthogonal projection of $r$ into $H_p^{xy}$.

   Then,
   $$
      |r r'| \leq |pr| \tan (\alpha)
             \leq \frac{2 \sqrt{2} |pq|}{10 m}
             \leq \frac{3 |pq|}{10 m}
             \leq \frac{3 |b_1|}{10 m}
             \leq \frac{4 |b_1|}{10 m + 20}
             =    4 \delta |b_1|,
   $$
   where we use the fact that $3 / (10 m) \leq 4 / (10 m + 20)$ if and
   only if $m \geq 6$, which is guaranteed.
   This implies that $H_p$ cannot intersect $\Bopt (S)$ since the
   vertical distance between $H_p$ and $H_p^{xy}$ (for such a point $r$)
   is less than $4 \delta |b_1|$ (four cells of $\M$),
   whereas the vertical distance between $H_p^{xy}$ and
   the bottom of $\Bopt (S)$ is at least $4 \delta |b_1|$.
   Similarly, the plane $H_q$ cannot intersect the box $\Bopt (S)$.
   The same argument shows that $H_p$ and $H_q$
   intersect only the vertical edges of
   $\B$ but do not intersect its top and bottom faces.

   Let $\A$ be the orthogonal projection of $\Bopt (S)$ into $H_p$
   (along $pq$), and let $\V$ be the product of $\A$ and $pq$
   (that is, a prism whose base is $\A$).
   It is easy to verify that $\Bopt (S) \subseteq \V \subseteq \B$.
   Obviously, the rectangle $R = \B \cap H_p^{xy}$ contains the set
   $\V \cap H_p^{xy}$.
   Moreover, according to \lemref{project:rects} the orthogonal
   projection of $R$ into $H_p$ is contained in a rectangle
   $R' \in H_p$ such that $\Area (R') \leq \Area (R)$.
   Since $\V \cap H_p^{xy} \subseteq R$, we also have $\A \subseteq R'$.
   Let $\hat{B}$ be the box with $R'$ as its base and the opposite face
   lying on $H_q$.
   Clearly, $S \subseteq \Bopt(S) \subseteq \V \subseteq \hat{B}$.
   Finally, we have
   $$
      \Vol(\hat{B}) =    |pq| \Area(R')
                    \leq |b_1| \Area(R)
                    \leq \Vol (\B)
                    \leq (1 + \eps) \, \Vol (\Bopt (S)).
   $$

   Since $\Vol (\hat{B}) \geq \Vol (\Bopt (S, \{pq\}))$, the latter box $\Bopt(S,\{pq\})$ is also a $(1 + \eps)$-approximation of $\Bopt (S)$.
\end{proof}

We complete this discussion by showing that $\vec{pq}$ is indeed ``short,'' namely, that $|pq|_\infty$ (measured by $\G$-grid units) is small.  Intuitively, this follows from the fact that the grid sizes of $\M$ and $\G$ are comparable up to a multiplicative constant.  By \lemref{obvious} we have that a copy of the unit box of $\M$ (that is, $\delta \B$), scaled down by a constant factor and translated, is contained by the unit box of $\G$ (that is, $c \eps B^{*}(S)$).  In particular, this implies that every grid box of $\M$ is covered by a constant number of grid boxes of $\G$.  In addition, the segment $pq$ is contained in $\cup_{i=5}^{20 m + 15} B^\M_{(0,0,i)}$, whose height (size along $z$) is $O (1 / \eps)$. Thus, the segment $pq$ can be covered by $O(1/\eps)$ grid boxes of $\G$. Hence, all the coordinates of $q-p$ (in $\G$ units) are $O (1 / \eps)$, where the constant of proportionality hidden in the big-Oh notation is $c$ (the constant used by the algorithm).  We thus establish the following theorem:

\begin{theorem}
   \thmlab{aprox:ext}

   Let $S$ be a set of $n$ points in $\Re^3$, and let $0 < \eps \leq 1$
   be a parameter.
   One can compute in $O (n \log n + n / \eps^3)$ time a bounding box
   $B (S)$ with $\Vol (B (S)) \leq (1 + \eps) \, \Vol (\Bopt (S))$.
\end{theorem}

The algorithm described above may be too slow to use in practice because the constant of proportionality hidden in the big-Oh notation (affected by the value of $c$) may be too large.  However, it suggests the heuristic of computing the bounding boxes $\Bopt (S, \{v\})$ induced by directions defined by grid points of $\G$ ``close'' to $o$. \thmref{aprox:ext} implies that the higher the bound on the ``length'' of $v$ is, the better the approximation is.

\section{Experimental Results}

\seclab{experiments}

We have implemented software that computes the exact 2-dimensional
minimum-area bounding rectangle of a planar point set (by computing its
convex hull and then applying a rotating-calipers algorithm).
Based on this tool we have implemented software that computes
the exact 3-dimensional minimum-volume bounding box of a spatial point set,
one of whose directions is given.
We have used the latter tool for implementing several approximation heuristics
for the minimum-volume bounding box, and report here on several of them.
The entire software was implemented in plain C and it runs under any
Unix-like operating system. The running times reported here were
measured in a Linux environment on a 200-MHz Pentium-Pro machine.
The software consists of about 1,500 lines of code.

It was easy to observe that for any bounding box $B (S)$,
one can always ``locally'' improve (decrease) the
volume of the box by projecting it into a plane perpendicular to one of
the directions of $B (S)$, followed by computing the minimum-area
bounding rectangle of the projected set in that plane, and by using this
rectangle as the base of an improving bounding box of $S$.
Our experimental results revealed many examples in which this procedure
converges to a local (but not a global) minimum-volume bounding box.
Nevertheless, this procedure improves  (by a small amount) the solutions
produced by the approximation algorithms described in this paper.
We performed each experiment twice, without and with this
solution-improving step.

Here are three examples (out of the many examples we experimented with)
of the performance of the heuristics.
\figref{examples}(a.1) shows
a rotated version of the set
\[
    S = \{(-1, -0.1, 0), (-1, 0.1, 0), (1, 0, -0.1), (1, 0, 0.1)\}.
\]
(The points are displayed as small triangles.)  \figref{examples}(a.2) and~\figref{examples}(a.3) show $B^* (S)$ (which is also optimal among all boxes aligned with some diagonal of $\CH (S)$), and the improved $B^* (S)$-G(20) (see below), respectively.  \figref{examples}(b) shows a set $S$ of 48 arbitrary points.  The diameter of the set is shown as a nearly-vertical line segment.  The figure shows $B^* (S)$ (the nearly-vertical box) and the minimum-volume bounding box of $S$ among those aligned with some hull diagonal.  \figref{examples}(c) shows the same types of boxes bounding a set of 100 points which were randomly and uniformly selected on the unit sphere.

\begin{figure*}
   \centering
   \begin{tabular}{cc}
      \multicolumn{2}{c}{\includegraphics[width=2.0in]{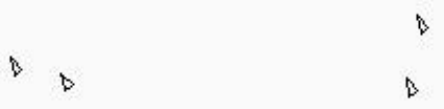}
         }
      \smallskip \\
      \multicolumn{2}{c}{(a.1) Points ($S$)} \\
     {\includegraphics[width=2.0in]{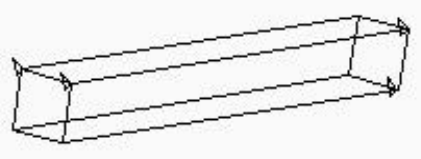}}
     &
       {\includegraphics[width=2.0in]{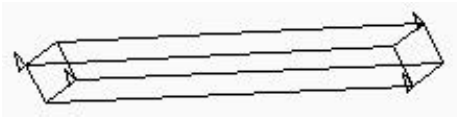}}%
     \\
      (a.2) $B^* (S)$ & (a.3) Improved $B^* (S)$-G(20) \smallskip \\
      \multicolumn{2}{c}{(a) 4 points} \medskip \\
      \includegraphics[width=2.1in]{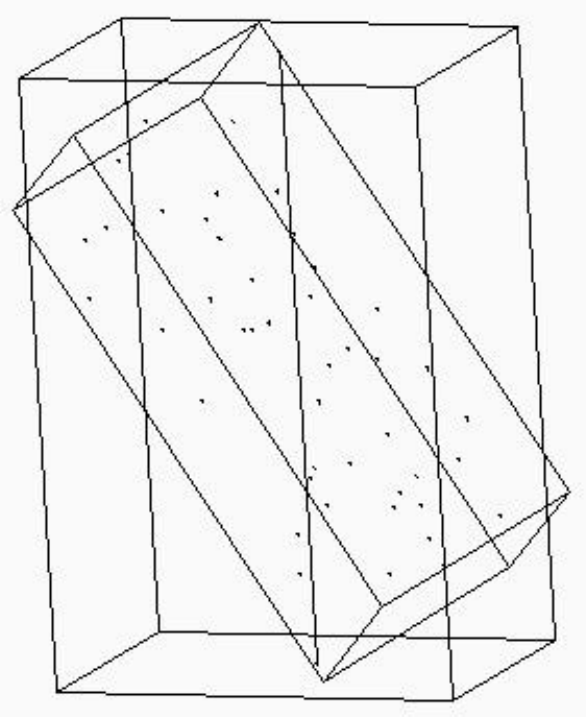} &
         \includegraphics[width=2.0in]{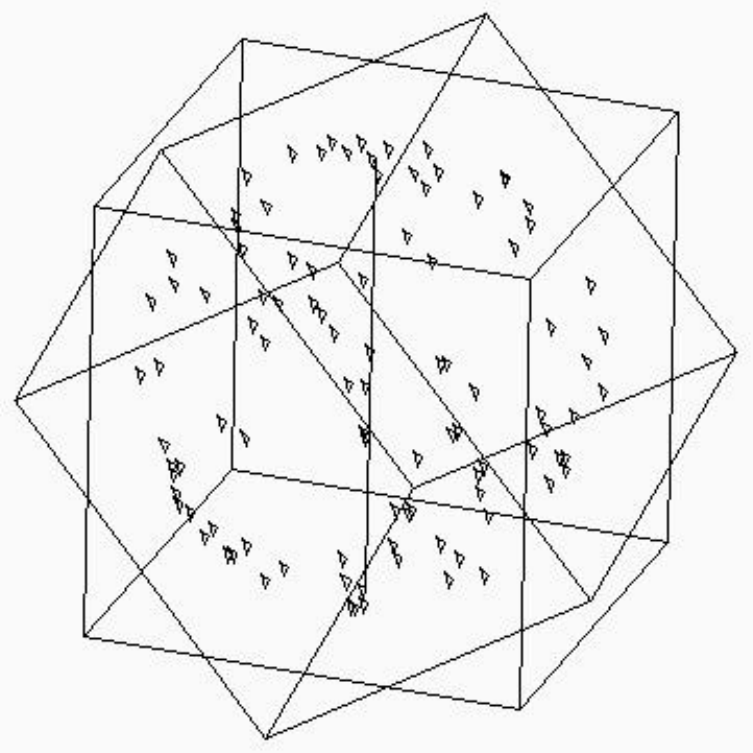} \\
      (b) 48 points & (c) 100 points
   \end{tabular}
   \caption{Bounding boxes of three spatial point sets}
   \figlab{examples}
\end{figure*}

\begin{table*}
   \centering
   \footnotesize
   \begin{tabular}{|c|l|l|r|r|c|c|}
      \hline
       & & & & \multicolumn{1}{c|}{Calls to} &
         \multicolumn{2}{c|}{Time} \\
      $|S|$ & Distribution & Box & \multicolumn{1}{c|}{Volume} &
         \verb+MVBB(v)+ & Per call & Total \\
      \hline
      ~~4 &           & $B^* (S)$ & 0.07980 & 1 & (unreliable) & 0 \\
          &           & All pairs & 0.07980 & 1 & (unreliable) & 0 \\
          &           & $B^* (S)$-G(2) & 0.03995 & 23 & (unreliable) & 0 \\
          &           & ~~~~(Improved) & 0.03995 & 3 & (unreliable) & 0 \\
          &           & $B^* (S)$-G(5) & 0.03995 & 339 & 118 $\mu$Sec &
                         0.04 Sec \\
          &           & ~~~~(Improved) & 0.03995 & 3 & (unreliable) & 0 \\
          &           & $B^* (S)$-G(10) & 0.03995 & 3107 & 113 $\mu$Sec &
                         0.35 Sec \\
          &           & ~~~~(Improved) & 0.03995 & 3 & (unreliable) & 0 \\
          &           & $B^* (S)$-G(20) & 0.03995 & 26019 & 113 $\mu$Sec &
                         2.95 Sec \\
          &           & ~~~~(Improved) & 0.03995 & 3 & (unreliable) & 0 \\
          &           & $xyz$-G(2) & 0.07974 & 23 & (unreliable) & 0 \\
          &           & ~~~~(Improved) & 0.05267 & 15 & (unreliable) & 0 \\
          &           & $xyz$-G(5) & 0.05674 & 339 & 118 $\mu$Sec &
                         0.04 Sec \\
          &           & ~~~~(Improved) & 0.04202 & 27 & (unreliable) &
                         0.01 Sec \\
          &           & $xyz$-G(10) & 0.05005 & 3107 & 119 $\mu$Sec &
                         0.37 Sec \\
          &           & ~~~~(Improved) & 0.04009 & 9 & (unreliable) & 0 \\
          &           & $xyz$-G(20) & 0.04082 & 26019 & 119 $\mu$Sec &
                         3.10 Sec \\
          &           & ~~~~(Improved) & 0.04082 & 3 & (unreliable) & 0 \\
      \hline
      ~48 & Arbitrary & $B^* (S)$ & 168.82 & 1 & (unreliable) & 0 \\
          &           & All pairs & 83.20 & 1,128 &
                            674 $\mu$Sec & 0.76 Sec \\
          &           & $B^* (S)$-G(2) & 87.11 & 23 &
                            (unreliable) & 0.02 Sec \\
          &           & ~~~~(Improved) & 83.24 & 18 &
                            (unreliable) & 0.01 \\
          &           & $B^* (S)$-G(5) & 84.13 & 339 &
                            678 $\mu$Sec & 0.23 Sec \\
          &           & ~~~~(Improved) & 83.39 & 12 &
                            (unreliable) & 0 \\
          &           & $B^* (S)$-G(10) & 83.28 & 3,107 &
                            679 $\mu$Sec & 2.11 Sec \\
          &           & ~~~~(Improved) & 83.18 & 9 &
                            (unreliable) & 0 \\
          &           & $B^* (S)$-G(20) & 83.28 & 26,019 &
                            677 $\mu$Sec & 17.61 Sec \\
          &           & ~~~~(Improved) & 83.18 & 9 &
                            (unreliable) & 0.01 Sec \\
          &           & $xyz$-G(2) & 83.22 & 23 &
                            (unreliable) & 0.02 Sec \\
          &           & ~~~~(Improved) & 83.20 & 6 &
                            (unreliable) & 0 \\
          &           & $xyz$-G(5) & 83.22 & 339 &
                            678 $\mu$Sec & 0.23 Sec \\
          &           & ~~~~(Improved) & 83.20 & 6 &
                            (unreliable) & 0 \\
          &           & $xyz$-G(10) & 83.22 & 3,107 &
                            653 $\mu$Sec & 0.23 Sec \\
          &           & ~~~~(Improved) & 83.20 & 6 &
                            (unreliable) & 0.01 Sec \\
          &           & $xyz$-G(20) & 83.22 & 26,019 &
                            658 $\mu$Sec & 17.13 Sec \\
          &           & ~~~~(Improved) & 83.11 & 6 &
                            (unreliable) & 0 \\
      \hline
      100 & Uniform   & $B^*$ & 7.333 & 1 &
                            (unreliable) & 0 \\
          & on a unit & All pairs & 6.422 & 4,950 &
                            1,596 $\mu$Sec & 7.99 Sec \\
          & sphere    & $B^* (S)$-G(2) & 6.688 & 23 &
                            (unreliable) & 0.04 Sec \\
          &           & ~~~~(Improved) & 6.601 & 15 &
                            (unreliable) & 0.02 \\
          &           & $B^* (S)$-G(5) & 6.446 & 339 &
                            1,622 $\mu$Sec & 0.55 Sec \\
          &           & ~~~~(Improved) & 6.420 & 18 &
                            (unreliable) & 0.03 Sec \\
          &           & $B^* (S)$-G(10) & 6.427 & 3,107 &
                            1,641 $\mu$Sec & 5.10 Sec \\
          &           & ~~~~(Improved) & 6.418 & 9 &
                            (unreliable) & 0.02 Sec \\
          &           & $B^* (S)$-G(20) & 6.421 & 26,019 &
                            1,625 $\mu$Sec & 42.27 Sec \\
          &           & ~~~~(Improved) & 6.421 & 3 &
                            (unreliable) & 0 \\
          &           & $xyz$-G(2) & 6.719 & 23 &
                            (unreliable) & 0.04 Sec \\
          &           & ~~~~(Improved) & 6.526 & 21 &
                            (unreliable) & 0.03 Sec \\
          &           & $xyz$-G(5) & 6.462 & 339 &
                            1,622 $\mu$Sec & 0.55 Sec \\
          &           & ~~~~(Improved) & 6.418 & 12 &
                            (unreliable) & 0.02 Sec \\
          &           & $xyz$-G(10) & 6.440 & 3,107 &
                            1,629 $\mu$Sec & 5.06 Sec \\
          &           & ~~~~(Improved) & 6.422 & 12 &
                            (unreliable) & 0.02 Sec \\
          &           & $xyz$-G(20) & 6.426 & 26,019 &
                            1,638 $\mu$Sec & 42.63 Sec \\
          &           & ~~~~(Improved) & 6.419 & 9 &
                            (unreliable) & 0.01 Sec \\
      \hline
   \end{tabular}
   \caption{Performance of the approximation heuristics}
   \tbllab{performance}
\end{table*}

\tblref{performance} shows the volumes of several bounding boxes of the three spatial sets, and the corresponding running times of our software.  The box $B^* (S)$ is the minimum-volume bounding box aligned with the diameter of the set $S$. The ``all-pairs'' box is obtained by minimizing the volume of all the boxes aligned with directions which connect some two points of $S$. The suffix ``-G($k$)'' stands for checking all the boxes aligned with directions obtained by connecting the origin $o$ with a grid point whose $L_\infty$ norm is at most $k$. The table reports results for $\Grid (B^* (S))$ and for the regular Cartesian grid.  The column entitled ``\verb+MVBB(v)+'' details the number of calls to the function that computes $\Bopt (S, \{v\})$.

Note that although the more laborious heuristics require much time to run,
some of the faster heuristics perform reasonably well in practice.
For example, the $B^* (S)$-G(5) variant combined with the improvement step
produces good approximating boxes and runs fast.
Moreover, this heuristic performs considerably better than the uniform-grid
heuristic for ``long and skinny'' sets of points.

\section{Conclusion}
\seclab{conclusion}

In this paper we present an efficient algorithm for approximating the
minimum-volume bounding box of a point set in $\Re^3$.
We also present a simpler algorithm which we implemented and
experimented with on numerous three-dimensional point sets.

Jeff Erickson pointed out in a personal communication that it was possible to reduce the $O (1 / \eps^{4.5})$ term in the running time of the first algorithm to $O (1 / \eps^3)$ (in the cost of adding an $O (n \log{n})$ term) by using Dudley's method~\cite{d-mescs-74,ahsv-aspcp-97}. The main idea is to scale down the space so as to transform $B^* (S)$ to a unit cube.  There one computes a $(c \eps)$-approximation (for a suitable constant $c$) of $\CH (S)$ and scales up this approximating polyhedron back to the original space.  Similarly to our argumentation in \secref{first-alg}, one can show that the minimum-volume bounding box of the scaled-up polyhedron is a $(1 + \eps)$-approximation of the minimum-volume bounding box of $S$. This version of the algorithm requires $O (n \log{n} + 1 / \eps^3)$ time.

The implementation of the algorithm described in this paper
is available online \cite{h-scpca-00}.  We conclude by
mentioning two open problems:

\begin{itemize}
\item Can one maintain dynamically a $(1 + \eps)$-approximation of the
      minimum-volume bounding box of a moving point set in $\Re^3$?
\item Can one compute efficiently and by a simple algorithm a
      $(1 + \eps)$-approximation of the minimum-volume bounding ellipsoid of a
      point set in $\Re^3$?
\end{itemize}

\section*{Acknowledgement}

The authors wish to thank Pankaj Agarwal, Jeff Erickson, and Micha
Sharir for helpful discussions concerning the problem studied in this
paper and related problems.
This work is part of the second author's Ph.D. thesis, prepared at
Tel-Aviv University under the supervision of Prof.\ Sharir.

\BibLatexMode{\printbibliography}

\end{document}